%% file: main.tex
\documentclass[copyright]{eptcs}
\providecommand{\event}{SR 2013} % Name of the event you are submitting to
\usepackage{breakurl}             % Not needed if you use pdflatex only.
\usepackage{color}

\usepackage{amssymb}
\usepackage{amsmath}
\usepackage{amsthm}
\usepackage{stmaryrd}

\usepackage[ruled]{algorithm2e}
\usepackage{algorithmic}

\usepackage{tikz}
\usetikzlibrary{arrows,positioning,backgrounds,calc,shapes,decorations}
\usetikzlibrary{decorations.pathmorphing}
\pgfdeclarelayer{foreground}
\pgfdeclarelayer{background}
\pgfsetlayers{background,main,foreground}

\input{macros.tex}

\title{Lossy Channel Games under Incomplete Information}
\author{Rayna Dimitrova
\institute{Saarland University, Germany}
\email{dimitrova@cs.uni-saarland.de}
\and
Bernd Finkbeiner
\institute{Saarland University, Germany}
\email{finkbeiner@cs.uni-saarland.de}
}
\def\titlerunning{Lossy Channel Games under Incomplete Information}
\def\authorrunning{R. Dimitrova, B. Finkbeiner}

\begin{document}
\sloppy 

\maketitle

\begin{abstract}
In this paper we investigate lossy channel games under incomplete information, where two players operate on a finite set of unbounded FIFO channels and one player, representing a system component under consideration operates under incomplete information, while the other player, representing the component's environment is allowed to lose messages from the channels. We argue that these games are a suitable model for synthesis of communication protocols where processes communicate over unreliable channels. We show that in the case of finite message alphabets, games with safety and reachability winning conditions are decidable and finite-state observation-based strategies for the component can be effectively computed. Undecidability for (weak) parity objectives follows from the undecidability of (weak) parity perfect information games where only one player can lose messages.
\end{abstract}

\section{Introduction}\label{Sec:intro}
\input{intro.tex}

\section{Lossy Channel Games under Incomplete Information}\label{Sec:prelim}
\input{preliminaries.tex}

\section{Algorithms for Solving Safety and Reachability Games}\label{Sec:safety_reachability}

\input{bqo-games.tex}

\input{safety.tex}

\input{reachability.tex}

\input{parity.tex}

\section{Conclusion}\label{Sec:conclusion}
\input{conclusion.tex}

\bigskip
\noindent
{\bf Acknowledgements} This work is partially supported by the DFG as part of SFB/TR 14 AVACS.

\bibliographystyle{eptcs}
\bibliography{main}

\end{document}

%% file: macros.tex
\theoremstyle{plain}
\newtheorem{theorem}{Theorem}

\newtheorem{proposition}{Proposition}

\theoremstyle{definition}
\newtheorem{definition}{Definition}

\theoremstyle{remark}
\newtheorem*{rem}{Remark}

\newcommand{\comment}[1]{}

\newcommand{\myparagraph}[1]{\smallskip\noindent\textbf{#1.}  }

\newcommand{\mc}[1]{\mathcal{#1}}

\newcommand{\true}{\ensuremath{\mathit{true}}}
\newcommand{\false}{\ensuremath{\mathit{false}}}

\newcommand{\acte}{\ensuremath{\Sigma_\exists}}
\newcommand{\Op}{\ensuremath{\mathit{Op}}}
\newcommand{\Gr}{\ensuremath{\mathit{Gr}}}

\newcommand{\nop}{\ensuremath{\mathit{nop}}}
\newcommand{\obs}{\ensuremath{\mathit{obs}}}
\newcommand{\enabled}{\ensuremath{\mathsf{Enabled}}}

\newcommand{\playerf}{\ensuremath{\mathit{Player}_\forall}}
\newcommand{\playere}{\ensuremath{\mathit{Player}_\exists}}
\newcommand{\transg}{\ensuremath{\rightarrow_g}}
\newcommand{\transgw}{\ensuremath{\Rightarrow_g}}
\newcommand{\actebot}{\ensuremath{\Sigma_\exists^\bot}}
\newcommand{\last}{\ensuremath{\mathsf{last}}}
\newcommand{\prefs}{\ensuremath{\mathsf{Prefs}}}
\newcommand{\prefse}{\ensuremath{\mathsf{Prefs}_\exists}}
\newcommand{\stre}{\ensuremath{f_\exists}}
\newcommand{\outcome}{\ensuremath\mathsf{Outcome}}
\newcommand{\safety}{\ensuremath\mathsf{Safety}}
\newcommand{\reach}{\ensuremath\mathsf{Reach}}
\newcommand{\err}{\ensuremath Err}
\newcommand{\goal}{\ensuremath Goal}
\newcommand{\Obs}{\ensuremath \mathit{Obs}}
\newcommand{\Act}{\ensuremath \mathit{Act}}
\newcommand{\pr}{\ensuremath \mathit{pr}}
\newcommand{\wpr}{\ensuremath \mathit{wpr}}
\newcommand{\parity}{\ensuremath \mathit{Parity}}
\newcommand{\wparity}{\ensuremath \mathit{WeakParity}}
\newcommand{\InfObs}{\ensuremath \mathit{InfObs}}
\newcommand{\States}{\ensuremath \mathit{States}}

\newcommand{\pfin}{\ensuremath\mc P_\mathsf{fin}} 
\newcommand{\ucsets}{\ensuremath\mc{U}_\obs}
\newcommand{\dcsets}{\ensuremath\mc{D}_\obs}
\newcommand{\dcsetsf}{\ensuremath\mc{D}_\obs^{\mathsf{fin}}}

\newcommand{\pre}{\ensuremath\mathsf{Pre}}
\newcommand{\post}{\ensuremath\mathsf{Post}}

\newcommand{\lsets}{\ensuremath\mathcal{L}}
\newcommand{\Min}{\ensuremath\mathsf{Min}}

\newcommand{\games}{\ensuremath\widetilde{\mc G}}
\newcommand{\idle}{\ensuremath \mathit{idle}}
\newcommand{\transgs}{\ensuremath{\widetilde\rightarrow_g}}

\newcommand{\win}{\ensuremath{\mathit{win}}}
\newcommand{\Children}{\ensuremath{\mathit{Children}}}

%% file: intro.tex
Lossy channel systems (LCSs), which are finite systems communicating via unbounded lossy FIFO channels, are used to model communication protocols such as link protocols, a canonical example of which is the Alternating Bit Protocol. The decidability of verification problems for LCSs has been well studied and a large number of works have been devoted to developing automatic analysis techniques. In the control and synthesis setting, where games are the natural computational model, this class of systems has not yet been so well investigated. In~\cite{AbdullaBd/08/MonotonicAndDownwardClosedGames}, Abdulla et al. establish decidability of two-player safety and reachability games where one (or both) player has downward-closed behavior (e.g., can lose messages), which subsumes games with lossy channels where one player (i.e., the environment) can lose messages. They, however, assume that the game is played under perfect information, which assumption disregards the fact that a process has no access to the local states of other processes or that it has only limited information about the contents of the channels. To the best of our knowledge, games under incomplete information where the players operate on unbounded unreliable channels have not been studied so far. 

We define \emph{lossy channel games under incomplete information} and show that in the case of finite message alphabets, games  with safety and reachability winning conditions are decidable and finite-state observation-based strategies for the player who has incomplete information can be effectively computed.

Algorithms for games under incomplete information carrying out an explicit knowledge based subset construction~\cite{Reif/84/TheComplexityOfTwoPlayerGamesOfIncompleteInformation} are not directly applicable to infinite-state games. Symbolic approaches~\cite{ChatterjeeDHR/06/AlgorithmsForOmegaRegularGamesOfIncompleteInformation} are effective for restricted classes of infinite-state games like discrete games on rectangular automata~\cite{DeWulfDR/06/ALatticeTheoryForSolvingGamesOfImperfectInformation}.
The symbolic algorithms that we present in this paper rely on the monotonicity of lossy channel systems w.r.t. the subword ordering, which is a well-quasi ordering (WQO). It is well known that upward and downward-closed sets of words used in the analysis of lossy channel systems can be effectively represented by finite sets of minimal elements and simple regular expressions~\cite{AbdullaCBJ/04/UsingForwardReachabilityAnalysisForVerificationOfLossyChannelSystems}, respectively.
Unsurprisingly, the procedures for solving lossy channel games under incomplete information that we develop manipulate sets of sets of states. Thus, our termination arguments rely on the fact that the subword ordering is in fact a better-quasi ordering (BQO)~\cite{Milner/85/BasicWqoAndBqoTheory,Nash-Williams/65/OnWellQuasiOrderingInfiniteTrees}, a stronger notion than WQO that is preserved by the powerset operation~\cite{Marcone/01/FineAnalysisOfTheQuasiOrderingsOnThePowerSet}.

%% file: preliminaries.tex
Lossy channel systems are asynchronous distributed systems composed of finitely many finite-state processes communicating through a finite set of unbounded FIFO channels that can nondeterministically lose messages. We consider \emph{partially specified lossy channel systems}, where the term partially specified refers to the fact that we consider a second ("friendly") type of nondeterminism, in addition to the ("hostile") one due to the model. More specifically, this second type of nondeterminism models \emph{unresolved implementation decisions} that can be resolved in a favorable way. We consider the case when these decisions are within a single process, and thus we can w.l.o.g.\ assume that the system consist of only two processes: the process under consideration and the parallel composition of the remaining processes.

\begin{definition}
A \emph{partially specified lossy channel system (LCS)} is a tuple $\mc L = (\mc A_0,\mc A_1, C, M, \Sigma_0,\Sigma_1,\Sigma_\exists)$, where for each \emph{process identifier} $p \in \{0,1\}$, $\mc A_p$ is a finite automaton describing the behavior of process $p$, $C$ is a finite set of \emph{channels}, $M$ is a finite set of \emph{messages}, $\Sigma = \Sigma_0 \dot\cup \Sigma_1$ is the union of the disjoint finite sets of \emph{transition labels} for the two processes, and $\Sigma_\exists \subseteq \Sigma_0$ is a subset of the labels of the \emph{partially specified process} $\mc A_0$. The automaton $\mc A_p = (Q_p,q_p^0,\delta_p)$ for a process $p$ consists of a finite set $Q_p$ of \emph{control locations}, an \emph{initial location} $q_p^0$ and a finite set $\delta$ of \emph{transitions} of the form $(q,a,\Gr,\Op,q')$, where $q,q'\in Q_p$,  $a \in \Sigma_p$, $\Gr: C \to \{\true,(=\epsilon),\in (m\cdot M^*)\mid m \in M\}$ and $Op: C \to \{!m,?m,\nop\mid m \in M\}$. 
Intuitively, the function $\Gr$ maps each channel to a guard, which can be an emptiness test, a test of the letter at the head of the channel or the trivial guard \emph{true}. The function $\Op$ gives the update operation for the respective channel, which is either a write, a read or $\nop$, which leaves the channel unchanged. 
\end{definition}

\myparagraph{Example}
Fig.\ref{fig:abp_partial} depicts a partially specified protocol consisting of two processes, {\sc Sender} and {\sc Receiver}, communicating over the unreliable channels $K$ and $L$. Process {\sc Sender} sends messages to {\sc Receiver} over channel $K$ and {\sc Receiver} acknowledges the receipt of a message using channel $L$. Note that we use guards that test channels for emptiness or test the first letter of their contents.
\input{example-abp}

The two processes are represented as nondeterministic finite-state automata. Process {\sc Sender} essentially runs the Alternating Bit Protocol. Process {\sc Receiver}, however, is only \emph{partially specified}: its alphabet of transition labels $\Sigma_0 = \{a_0,a_1,b_0,b_1,u\}$ is partitioned according to the unresolved decisions in the process specification: The subset $\Sigma_\exists = \{b_0,b_1\}$  of controllable transition labels specifies the unresolved implementation decisions, namely what bit to be sent on channel $L$ at location $1$. 

The property that the protocol must satisfy is encoded as the unreachability of location $4$ in process {\sc Sender}. However, the automata can easily be augmented (with an extra channel and an error location in process {\sc Receiver}) in a way that the error location is in process {\sc Receiver}. The property states that:
\begin{enumerate}\setlength{\itemsep}{-.5mm}
 \item the receiver does not acknowledge messages that have not been sent, that is in location $2$ in {\sc Sender} the language of $L$ is $0^*$ and in location $0$ in {\sc Sender} the language of $L$ is $1^*$,
 \item once all messages and acknowledgements trailing from previous phases have been consumed (or lost), the number of delayed acknowledgements the receiver can send is bounded by one.
\end{enumerate}

\bigskip 

A \emph{configuration} $\gamma = (q_0,q_1,w)$ of $\mc L$ is a tuple of the locations of the two processes and a function $w : C \to M^*$ that maps each channel to its contents. The \emph{initial configuration} of $\mc L$ is $\gamma^0 = (q_0^0,q_1^0,\epsilon)$, where $\epsilon(c) = \epsilon$ for each $c \in C$. The set of possible channel valuations is $W = \{w\mid w: C \to M^*\}$. 

The \emph{strong labeled transition relation} $\rightarrow \subseteq (Q_0 \times Q_1 \times W) \times \Sigma \times (Q_0 \times Q_1 \times W)$ of $\mc L$ consists of all tuples $((q_0,q_1,w),a,(q_0',q_1',w'))$ (denoted $(q_0,q_1,w)\stackrel{a}{\rightarrow}(q_0',q_1',w')$) such that if $a \in \Sigma_p$, then $q_{1-p}' = q_{1-p}$ and there is a transition $(q_p,a,\Gr,\Op,q_p') \in \delta$ such that for each $c \in C$ all of the following conditions hold: (1) if $\Gr(c) = (\in m \cdot M^*)$ then $w(c) \in m \cdot M^*$, (2) if $\Gr(c) = (=\epsilon)$ then $w(c) = \epsilon$, (3) if $\Op(c) = !m$, then $w'(c) = w(c) \cdot m$, (4) if $\Op(c) = ?m$, then $m \cdot w'(c) = w(c)$, and (5) if $\Op(c) = \nop$, then $w'(c) = w(c)$.  

Let $\preceq$ denote the (not necessarily contiguous) subword relation on $M^*$ and let us define its extension to elements of $W$ as follows: $w_1 \preceq w_2$ for $w_1,w_2 \in W$ iff $w_1(c) \preceq w_2(c)$ for every $c \in C$. 

The \emph{weak labeled transition relation} $\Rightarrow \subseteq (Q_0 \times Q_1 \times W) \times \Sigma \times (Q_0 \times Q_1 \times W)$ for $\mc L$ is defined as follows: $(q_0,q_1,w) \stackrel{a}{\Rightarrow} (q_0',q_1',w')$ iff there exist $w_1$ and $w_2$ such that $w_1 \preceq w$ and $w' \preceq w_2$ and $(q_0,q_1,w_1) \stackrel{a}{\rightarrow} (q_0', q_1',w_2)$, i.e., the channels can lose messages before and after the actual transition.

\begin{definition}[LC-game structure with incomplete information]\label{def:game-structure}
Let $\mc L = (\mc A_0,\mc A_1, C, M, \Sigma_0,\Sigma_1,\Sigma_\exists)$ be a partially specified LCS, and $C_\obs \subseteq C$ be a set of \emph{observable channels} that includes the set of all channels occurring in guards or read operations in $\mc A_0$. The \emph{lossy channel game structure with incomplete information} for $\mc L$ and $C_\obs$ is $\mc G(\mc L,C_\obs) = (S,I,\transg,C,M,\Sigma_0,\Sigma_1,\acte,C_\obs)$, where:
\begin{itemize} 
 \item The set of \emph{states} of $\mc G$ is $S  = \{0,1\} \times Q_0 \times Q_1 \times W$. The first component $p$ of a state $(p,q_0,q_1,w)$ identifies the process to be executed and the remaining ones encode the current configuration of $\mc L$. 
The set of initial states of $\mc G$ is $I = \{(p,q_0,q_1,w) \mid p \in \{0,1\},\ q_0 = q_0^0,\ q_1 = q_1^0,\ w = \epsilon\}$.

\item The labeled transition relations $\transg \subseteq S \times \Sigma \times S$ and $\transgw \subseteq S \times \Sigma \times S$ of $\mc G$ are defined as follows: for states $s = (p,q_0,q_1,w)$ and $s' = (p',q_0',q_1',w')$ and $a \in \Sigma$ we have
$s \stackrel{a}{\transg} s'$ iff $a \in \Sigma_p$ and $(q_0,q_1,w) \stackrel{a}{\rightarrow} (q_0',q_1', w')$, 
and we have $s \stackrel{a}{\transgw} s'$ iff $a \in \Sigma_p$ and $(q_0,q_1,w) \stackrel{a}{\Rightarrow} (q_0',q_1', w')$.

\end{itemize}
\end{definition}

\begin{rem}
The first component of states in $S$ is used to model the interleaving semantics and is updated nondeterministically in the transition relation $\transg$ (and $\transgw$). For simplicity, in Definition~\ref{def:game-structure} we do not make any assumptions about the nondeterministic choice of which process to be executed. One natural assumption one might want to make is that the selected process must have at least one transition enabled in the current state. This and other restrictions can be easily imposed in the above model.
\end{rem}

For the rest of the paper, let $\mc G = \mc G(\mc L,C_\obs) = (S,I,\transg,C,M,\Sigma_0,\Sigma_1,\acte,C_\obs)$ be the LC-game structure with incomplete information for a partially specified LCS $\mc L$ and observable channels $C_\obs$.

$\playere$ plays the game under incomplete information, observing only certain components of the current state of the game. Let $H_\obs = C_\obs \to (M \cup \{\epsilon\})$ and $\Obs = \{0,1\}\times Q_0\times H_\obs$. The \emph{observation function} $\obs: S \to \Obs$ maps each state $s = (p,q_0,q_1,w)$ in $\mc G$ to the tuple $\obs(s) = (p,q_0,h)$ of state components observed by $\playere$, where for each $c \in C_\obs$, if $p = 1$, then $h(c) = \epsilon$ and otherwise if $w(c) = \epsilon$, then $h(c) = \epsilon$ and if $w(c) = m\cdot w'$ for some $m \in M$ and $w' \in M^*$, then $h(c) = m$. That is, when $p = 0$ we have for $c \in C_\obs$ that $h(c)$ is the letter at the head of $w(c)$, when $c$ is not empty. For $o \in \Obs$, we denote with $\States(o) = \{s \in S\mid\obs(s) = o\}$ the set of states whose observation is $o$.

Let $S_0 = \{(p,q_0,q_1,w) \in S \mid p = 0\}$ be the states where process $0$ is to be executed and $S_1 = S \setminus S_0$.

The game $\mc G$ is played by $\playere$ and $\playerf$ who build up a play $s_0a_0^\exists a_0s_1 a_1^\exists a_1\ldots$, which is sequence of alternating states in $S$, labels in $\actebot=\acte\cup\{\bot\}$ and labels in $\Sigma$, starting with a state $s_0 \in I$. Each time the current state is in $S_0$, $\playere$ has to choose a label from the set $\acte\cup\{\bot\}$, that is either a label from $\acte$ of a transition enabled in the current state, or can be the special element $\bot$ in case no transition with label in $\acte$ is enabled or if there exists an enabled transition with label from $\Sigma_0 \setminus \acte$.

Let $\enabled(s) = \{a \in \Sigma_0 \mid \exists s'.\ s \stackrel{a}{\transg} s'\}$. Note that for states $s_1,s_2 \in S_0$ with $\obs(s_1) = \obs(s_2) = o$ it holds that $\enabled(s_1) = \enabled(s_2)$, and, abusing notation, we denote this set with $\enabled(o)$.  

For an observation $o = (0,q_0,h)$, the set $\Act_\exists(o) = (\enabled(o) \cap \Sigma_\exists) \cup \{\bot \mid \enabled(o) \cap \acte = \emptyset \text{ or } \enabled(o) \cap (\Sigma_0\setminus\acte) \not= \emptyset\}$ consists of the transition labels that $\playere$ can choose in a set $s \in S_0$ with $\obs(s) = o$. For a label $a^\exists \in \actebot$, the set $\Act_\forall(o,a^\exists) = (\{a^\exists\} \cap \acte) \cup (\enabled(o)\setminus\acte)$ consists of the transition labels which $\playerf$ can choose when the current choice of $\playere$ is $a^\exists$.

The play is built by $\playerf$ respecting the choices of $\playere$ and the transition relation $\transgw$. When $s_i \in S_0$, then $a_i^\exists \in \Act_\exists(\obs(s_i))$ is the transition label chosen by $\playere$ after the play prefix $s_0 a_0^\exists a_0s_1 a_1^\exists a_1\ldots a_{i-1}^\exists a_{i-1} s_i$ and $a_i \in \Act_\forall(\obs(s_i),a_i^\exists)$. After $\playere$ has made his choice, $\playerf$ resolves the remaining nondeterminism by choosing $a_i$ and the successor state $s_{i+1}$ to extend the play.

A \emph{play} in $\mc G$ is a sequence $\pi = s_0a_0^\exists a_0s_1 a_1^\exists a_1 s_1\ldots \in (S \cdot (\actebot \cdot \Sigma \cdot S)^* \cup S \cdot (\actebot\cdot\Sigma\cdot S)^\omega)$ such that $s_0 \in I$, for every $i\geq0$ it holds that $s_i \stackrel{a_i}{\transgw} s_{i+1}$, and
if $s_i\in S_1$, then $a^\exists = \bot$, and if $s_i \in S_0$ then $a_i^\exists \in \Act_\exists(\obs(s_i))$ and $a_i \in \Act_\forall(\obs(s_i),a_i^\exists)$. A play $\pi$ is finite iff $\last(\pi)$ has no successor in $\mc G$, where $\last(\pi) \in S$ is the last element of $\pi$. The set $\prefs(\mc G) \subseteq S \cdot (\actebot\cdot\Sigma\cdot S)^*$ consists of the finite prefixes of plays in $\mc G$, and we denote with $\prefse(\mc G) = \{ \pi \in \prefs(\mc G) \mid \last(\pi) \in S_0\}$ the set of prefixes ending in $S_0$.

A \emph{strategy for $\playere$} is a total function $\stre : \prefse(\mc G) \to \actebot$ such that $\stre(\pi) \in \Act_\exists(\obs(\last(\pi)))$.

The \emph{outcome of a strategy $\stre$} is the set of plays $\outcome(\stre)$ such that $\pi = s_0a_0^\exists a_0s_1 a_1^\exists a_1\ldots\in \outcome(\stre)$ iff for every $i \geq 0$ with $s_i \in S_0$ it holds that $a_i^\exists = \stre(s_0a_0^\exists a_0s_1 a_1^\exists a_1\ldots s_{i})$.

We define a function $\obs^+ : \prefse(\mc G) \to (\Obs \cdot \Sigma_0)^*\cdot \Obs$ that maps a prefix in $\prefse(\mc G)$ to the sequence of state and action observations made by $\playere$:
$\obs^+(s_0a_0^\exists a_0s_1 a_1^\exists a_1\ldots s_{n}) = \obs'(s_0)\cdot\obs'(a_0)\cdot\obs'(s_1)\cdot\obs'(a_1)\ldots \cdot\obs'(s_{n})$, where
for $s \in S$, we define $\obs'(s) = \obs(s)$ if $s \in S_0$ and $\obs'(s) = \epsilon$ otherwise, and for $a \in \Sigma$ we define $\obs'(a) = a$ if $a \in \Sigma_0$ and $\obs'(a) = \epsilon$ otherwise.
 
We call a strategy $\stre$ for $\playere$ \emph{$\obs^+$-consistent} if for every pair of prefixes $\pi_1$ and $\pi_2$ in $\prefse(\mc G)$ for which $\obs^+(\pi_1)  = \obs^+(\pi_2)$ holds, it also holds that $\stre(\pi_1) = \stre(\pi_2)$.

We are interested in \emph{finite-state} strategies for $\playere$, that is, strategies that can be implemented as finite automata. A finite state $\obs^+$-consistent strategy for $\playere$ in $\mc G$ is one that can be represented as a finite automaton $\mc M_s = (Q_s,q_s^0,(Q_0\times H_\obs)\times(\actebot\times\Sigma_0),\rho)$ with alphabet $(Q_0 \times H_\obs)\times(\actebot\times\Sigma_0)$, whose transition relation $\rho \subseteq (Q_s \times ((Q_0\times H_\obs)\times (\actebot\times\Sigma_0)) \times Q_s)$ has the following properties:
\begin{itemize} \setlength{\itemsep}{-.5mm} 
 \item[$(i)$] for each $q \in Q_s$, $o \in Q_0\times H_\obs$, $a^\exists \in \actebot$, $a \in \Sigma_0$, and $q_1',q_2' \in Q_s$, it holds that if $(q,(o,(a^\exists,a)),q_1') \in \rho$ and $(q,(o,(a^\exists,a)),q_2') \in \rho$, then $q_1' = q_2'$ (i.e., the transition relation $\rho$ is deterministic),
 \item[$(ii)$] for each $q \in Q_s$ and $o \in Q_0\times H_\obs$ there exist $a^\exists \in \actebot$, $a \in \Sigma_0$, $q' \in Q_s$ with $(q,(o,(a^\exists,a)),q') \in \rho$, 
 \item[$(iii)$] if $(q,(o,(a^\exists,a_1)),q'_1) \in \rho$ and $a_2 \in \Act_\forall((0,o),a^\exists)$, then $(q,(o,(a^\exists,a_2)),q'_2) \in \rho$ for some $q_2' \in Q_s$,
 \item[$(iv)$] if $(q,(o,(a_1^\exists,a_1)),q'_1) \in \rho$ and $(q,(o,(a_2^\exists,a_2)),q'_2) \in \rho$, then $a^\exists_1 = a^\exists_2$.
\end{itemize}

The automaton $\mc M_s$ defines an $\obs^+$-consistent strategy $\stre$ for $\playere$. According to the properties of $\mc M_s$, for each $\pi \in \prefse(\mc G)$ with $\obs^+(\pi) = o_0 a_0 o_1 a_1\ldots o_{n-1}a_{n-1}o_n$ there exists a unique sequence $a_0^\exists a_1^\exists a_{n-1}^\exists\in {\actebot}^n$ such that there is a run of $\mc M_s$ (also unique) on the word $o_0 a_0^\exists a_0 o_1 a_1^\exists a_1 \ldots o_{n-1} a_{n-1}^\exists a_{n-1}$. Let $q$ be the last state of this run. We then define $\stre(\pi) = a^\exists$, where  $a^\exists \in \actebot$ is the unique label that exists by conditions $(ii)$ and $(iv)$ such that there are $a \in \Sigma_0$ and $q \in Q_s$ such that $(q,(o_n,(a^\exists,a)),q') \in \rho$.

We now turn to the definition of winning conditions in LC-games under incomplete information. We consider \emph{safety} and \emph{reachability} winning conditions for $\playere$ defined by visible sets of states in $\mc G$. A set $T \subseteq S$ is \emph{visible} iff for every $s \in T$ and every $s' \in S$ with $\obs(s') = \obs(s)$ it holds that $s' \in T$.

A \emph{safety LC-game under incomplete information} $\safety(\mc G,\err)$ is defined by a LC-game structure with incomplete information $\mc G$ and a visible set $\err$ of error states that $\playere$ must avoid. A strategy $\stre$ for $\playere$ is \emph{winning} in $\safety(\mc G,\err)$ iff 
no play in $\outcome(\stre)$ visits a state in $\err$. 

Note that according to this definition, $\playere$ wins finite plays that do not reach an error state. If we want to ensure that plays reaching a state in $\mc G$ that corresponds to a deadlock in $\mc L$ are not winning for $\playere$, we can easily achieve this by appropriately instrumenting $\mc L$ and $\err$. 

A \emph{reachability LC-game under incomplete information} $\reach(\mc G,\goal)$ is defined by a LC-game structure with incomplete information $\mc G$ and a visible set $\goal$ of goal states that $\playere$ must reach. A strategy $\stre$ for $\playere$ is \emph{winning} in $\reach(\mc G,\goal)$ iff each play in $\outcome(\stre)$ visits a state in $\goal$. 

\begin{rem}
The definition of visible sets allows that $\err \cap S_1 \not = \emptyset$ and $\goal \cap S_1 \not = \emptyset$. Thus, our definition of visible objectives does not require that for each pair of plays $\pi_1$ and $\pi_2$ with $\obs^+(\pi_1) = \obs^+(\pi_2)$ (where $\obs^+$ is defined for plays analogously to prefixes) it holds that $\playere$ wins $\pi_1$ iff he wins $\pi_2$. For the algorithms, which we present in the next section, for solving safety and reachability LC-games under incomplete information, the objective for $\playere$ does not have to satisfy this condition.
\end{rem}

%% file: example-abp.tex
\definecolor{gray}{rgb}{0.5 0.5 0.5}
\definecolor{lightgray}{rgb}{0.90 0.90 0.90}

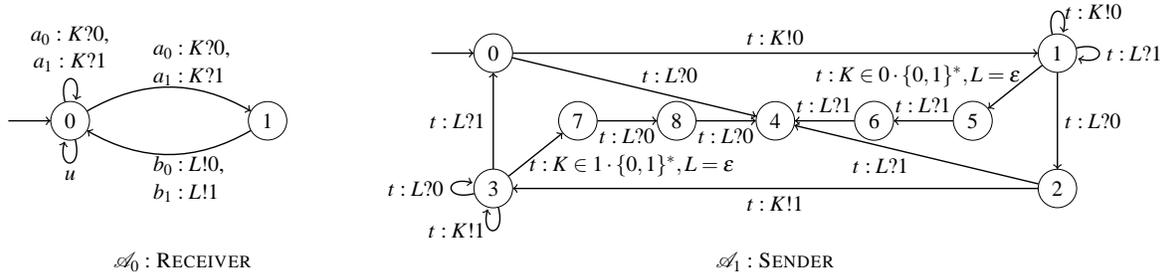
\begin{figure}
\centering
\scalebox{0.75}{
\begin{tikzpicture}

\tikzstyle{state} = [circle,draw=black,
		     minimum size = 1.5em,line width = .2pt]

\node (r0) [state] {$0$};
\node (r1) [state,right of = r0,xshift=2.5cm] {$1$};

\node (s4) [state,right of = r1,xshift=8cm] {$4$};
\node (s0) at ($(s4.center) + (-5cm,1.2cm)$)  [state,anchor=center] {$0$};
\node (s1) at ($(s4.center) + (5cm,1.2cm)$)   [state,anchor=center] {$1$};
\node (s2) at ($(s4.center) + (5cm,-1.2cm)$) [state,anchor=center] {$2$};
\node (s3) at ($(s4.center) + (-5cm,-1.2cm)$)  [state,anchor=center] {$3$};
\node (s11) at ($(s4.center) + (3.5cm,0cm)$) [state,anchor=center] {$5$};
\node (s12) at ($(s4.center) + (1.75cm,0cm)$) [state,anchor=center] {$6$};
\node (s31) at ($(s4.center) + (-3.5cm,-0cm)$) [state,anchor=center] {$7$};
\node (s32) at ($(s4.center) + (-1.75cm,-0cm)$) [state,anchor=center] {$8$};

\node (receiver) at ($(r0.center) + (2cm,-2.5cm)$) [anchor=center] {$\mc A_0:$ \sc Receiver};
\node (sender) at ($(s4.center) + (0,-2.5cm)$) [anchor=center] {$\mc A_1:$ \sc Sender};

\draw[->,draw = black,line width = .7pt] ($(r0.west) + (-.75,0)$) to  (r0.west);
\path[->,draw = black,line width = .7pt] 
(r0) edge [loop above] node[above,text width = 2cm, text centered]{$a_0:K?0$,\\$a_1:K?1$} (r0)
(r0) edge [loop below] node[below,text width = 2cm, text centered]{$u$} (r0)
(r0) edge [bend left] node[above,text width = 2cm,near end]{$a_0:K?0$,\\$a_1:K?1$} (r1) 
(r1) edge [bend left] node[below,text width = 2cm,near start]{$b_0:L!0$,\\$b_1:L!1$} (r0);

\draw[->,draw = black,line width = .7pt] ($(s0.west) + (-.75,0)$) to  (s0.west);
\path[->,draw = black,line width = .7pt] 
 (s0) edge node[above]{$t:K!0$} (s1)
 (s1) edge node[right]{$t:L?0$} (s2)
 (s2) edge node[below]{$t:K!1$} (s3)
 (s3) edge node[left]{$t:L?1$} (s0)
 (s1) edge [loop above] node[right] {$t:K!0$} (s1)
 (s1) edge [loop right] node[right] {$t:L?1$} (s1)
 (s3) edge [loop below] node[left] {$t:K!1$} (s3)
 (s3) edge [loop left]  node[left]  {$t:L?0$} (s0)
 (s0) edge node[right, yshift=.2cm]  {$t:L?0$} (s4)
 (s2) edge node[left, yshift=-.2cm]  {$t:L?1$} (s4)
 (s1) edge node[left,near start]  {$t:K\in0\cdot\{0,1\}^*,L=\epsilon$} (s11)
 (s11) edge node[above]  {$t:L?1$} (s12)
 (s12) edge node[above]  {$t:L?1$} (s4)
 (s3) edge node[right,near start]  {$t:K\in1\cdot\{0,1\}^*,L=\epsilon$} (s31)
 (s31) edge node[below]  {$t:L?0$} (s32)
 (s32) edge node[below]  {$t:L?0$} (s4)
;
 
\end{tikzpicture}
}
\caption{A communication protocol with partially specified {\sc Receiver} process. For process {\sc Receiver} we have $\Sigma_0 = \{a_0,a_1,b_0,b_1,u\}$ and $\Sigma_\exists = \{b_0,b_1\}$. The property that the implementation must satisfy is that location $4$ in {\sc Sender} is not reachable, i.e., the receiver does not acknowledge messages that have not been sent, and once all messages and acknowledgements from previous phases have been consumed, the receiver can only send one delayed acknowledgement. Note that by using an extra channel and an extra location in process {\sc Receiver} we can ensure that the error location is in process {\sc Receiver}.} 
\label{fig:abp_partial}
\end{figure}

%% file: bqo-games.tex
\myparagraph{Better-Quasi Orderings}
The subword ordering $\preceq$ on $M^*$ is a WQO (and so is the ordering $\preceq$ on $W$ defined earlier). That means, it is a reflexive and transitive relation such that for every infinite sequence $w_0,w_1,\ldots$ of elements of $M^*$ there exist indices $0\leq i < j$ such that $w_i \preceq w_j$. 

The subword ordering (as well as other WQOs commonly used in verification) is in fact also a BQO, and so is the ordering on $W$. Hence they are preserved by the powerset operation. Here we omit the precise definition of BQOs since it is rather technical and it is not necessary for the presentation of our results. When needed, we recall its properties relevant for our arguments.

We extend $\preceq$ to a BQO $\preceq$ on the set $S$ of states in $\mc G$ in the following way: for $s = (p,q_0,q_1,w) \in S$ and $s' = (p',q_0',q_1',w') \in S$, we have $s \preceq s'$ iff $p = p'$, $q_0 = q_0'$, $q_1 = q_1'$, $\obs(s) = \obs(s')$ and $w \preceq w'$.

A set $T \subseteq S$ is \emph{upward-closed} (respectively \emph{downward-closed}) iff for every $s \in T$ and every $s' \in S$ with $s \preceq s'$ (respectively $s' \preceq s$) it holds that $s' \in T$. The upward-closure of a set $T \subseteq S$ is $T\uparrow = \{s' \in S \mid \exists s.\ s \in T \text{ and } s \preceq s'\}$. 
For each upward (respectively downward) closed set $T \subseteq S$ and $o \in \Obs$, the set $T' = \{s \in T \mid \obs(s) = o\}$ is also upward (respectively downward) closed. We let $\ucsets(S) = \{u \subseteq S \mid u \not = \emptyset,\ u = u\uparrow \text{ and }\exists o \in \Obs.\forall s \in u.\ \obs(s) = o\}$ and for $u \in \ucsets(S)$ we define $\obs(u)$ in the obvious way. The set $\dcsets(S)$ and $\obs:\dcsets(S)\to\Obs$ are defined analogously, requiring that the elements are downward-closed instead of upward-closed. $\dcsetsf(S)$ is the set of finite sets in $\dcsets(S)$.

The transition relation $\transgw$ enjoys the following property: if $s\stackrel{a}{\transgw}s'$ and $s \preceq s''$, then $s''\stackrel{a}{\transgw} s'$. Thus, the set of predecessors w.r.t. some $a \in \Sigma$ of any set of states is upward-closed. For LCSs the set of successors w.r.t. some $a \in \Sigma$ of any set of states is a downward-closed set.

Let $\pre : \mc P(S) \times \Sigma \to \mc P(S)$ be the function defined as $\pre(T,a) = \{s \in S \mid \exists s' \in T.\ s \stackrel{a}{\transgw} s'\}$ and let $\post : \mc P(S) \times \Sigma \to \mc P(S)$ be the function defined as $\post(T,a) = \{s \in S \mid \exists s' \in T.\ s' \stackrel{a}{\transgw} s\}.$ As recalled above, for each $T \subseteq S$ and each $a \in \Sigma$, $\pre(T,a)$ is upward-closed and $\post(T,a)$ is downward-closed.

We define the functions $\pre_0 : \ucsets(S) \times \Sigma_0  \to \pfin(\ucsets(S))$ and $\pre_1 : \ucsets(S) \to \mc \pfin(\ucsets(S))$ that map a set $u \in \ucsets(S)$ to a finite set of upward-closed sets that partition the respective set of predecessors of $u$ according to the observations $\playere$ makes. Formally, $\pre_0(u,a) = \{u' \in \ucsets(S) \mid \exists o \in \Obs.\ u' = \pre(u,a) \cap \States(o)\}$ and $\pre_1(u) = \{u' \in \ucsets(S) \mid \exists o \in \Obs.\ u' = (\bigcup_{a \in \Sigma_1}\pre(u,a)) \cap \States(o)\}$. Similarly, using the function $\post$ above, we can define the successor functions $\post_0 : \dcsets(S) \times \Sigma_0 \to \pfin(\dcsets(S))$ and $\post_1 : \dcsets(S) \to \pfin(\dcsets(S))$.
Since the transition relation of $\mc G$ has finite branching, if $d \in \dcsetsf(S)$ then $d' \in \dcsetsf(S)$ for $d' \in \post_0(d,a)$ or $d' \in \post_1(d)$.

When analyzing LCSs, upward-closed sets are typically represented by their \emph{finite sets of minimal elements}, and downward-closed sets are represented by \emph{simple regular expressions}. These representations can be extended to obtain finite representations of elements of $\ucsets(S)$ and $\dcsets(S)$. By the definition of $\preceq$ on $S$, each visible set of states is upward-closed, and hence, the sets $\err$ and $\goal$ in safety and reachability games are finitely representable. In the rest, we assume that they are represented such a way.

Our termination arguments rely on the following property: For every BQO $\preceq$ on a set $X$, the superset relation $\supseteq$ is a BQO on the set of upward-closed sets in $\mc P(X)$ and the subset relation $\subseteq$ is a BQO on the set of downward-closed sets. This implies that $\supseteq$ is a BQO on $\ucsets(S)$ and that $\subseteq$ is a BQO on $\dcsets(S)$.

%% file: safety.tex
\myparagraph{LC-games under incomplete information with safety objectives}
We describe a decision procedure for safety LC-games under incomplete information which is based on a backward fixpoint computation.

Each step in the fixpoint computation corresponds to a step in the game, which is not necessarily observable by $\playere$. Thus, this construction is correct w.r.t.\  $\playere$ strategies that are $\widetilde\obs$-consistent, where, intuitively, the function $\widetilde\obs$ maps a prefix to a sequence that includes also the (trivial) observations of $S_1$ states, and $\widetilde\obs$-consistency is defined analogously to $\obs^+$-consistency. To avoid this problem, our algorithm performs the fixpoint computation on a LC-game structure with incomplete information $\games$ obtained from $\mc G$ by adding an \emph{idle transition} for process $1$. This game structure has the following property: $\playere$ has an $\obs^+$-consistent winning strategy in the game $\safety(\mc G,\err)$ iff $\playere$ has an $\widetilde\obs$-consistent winning strategy in $\safety(\games,\err)$, which yields correctness of the algorithm.

Formally, the function $\widetilde\obs: \prefse(\mc G) \to (\Obs^* \cdot \Sigma_0)^*\cdot \Obs$ is defined as:
$\widetilde\obs(s_0a_0^\exists a_0\ldots s_{n}) = \obs(s_0)\cdot\obs'(a_0)\cdot\ldots \cdot\obs(s_{n})$. The game structure $\games$ is the tuple $\games = (S,I,\transgs,C,M,\Sigma_0,\widetilde{\Sigma_1},\acte,C_\obs)$ where $\widetilde{\Sigma_1} = \Sigma_1 \cup\{\idle\}$ and $\idle \not \in \Sigma$, and $\transgs =\ \transg \cup \{((1,q_0,q_1,w),\idle,(p',q_0,q_1,w)) \mid p' \in \{0,1\}\}$. 

We define the set $\lsets(S)$ for $S$ as $\lsets(S) = \{l \in \pfin(\ucsets(S))\mid l \not= \emptyset \text{ and } \exists o \in \Obs.\forall u \in l.\ \obs(u)  = o\}$ and define $\obs(l)$ for each $l \in \lsets(S)$ in the obvious way. We provide a fixpoint-based  algorithm that computes a set $B \subseteq\lsets(S)$ such that each $l \in B$ has the following property: if $K \subseteq S$ is the set of states that the game can be currently in according to $\playere$'s knowledge and $K \cap u \not = \emptyset$ for every $u \in l$, then $\playere$ cannot win when his knowledge is $K$. Considering the set $I$ of initial states, if for some $l \in B$ it holds that $I \cap u \not = \emptyset$ for all $u \in l$, then $\playere$ has no $\obs^+$-consistent winning strategy in $\safety(\mc G,\err)$.

Our procedure computes a sequence $B_0 \subseteq B_1 \subseteq B_2\ldots$ of finite subsets of $\lsets(S)$. The computation starts with the set $B_0 = \{\{Err \cap \States(o)\} \mid o\in \Obs\}$. For $i \geq 0$, we let $B_{i+1} = B_i \cup  N_{i+1}$, where the set $N_{i+1}$ of new elements is computed based on $B_i$ and is the smallest set that contains each $l \in \lsets(S)$ which is such that $l \subseteq \bigcup_{l' \in B_i,u' \in l'}((\bigcup_{a \in \Sigma_0}\pre_{0}(u',a)) \cup \pre_{1}(u'))$ and:
\begin{itemize} \setlength{\itemsep}{-.5mm} 
 \item if $l \in \mc P(\mc P (S_0))$ then for every possible choice $a^\exists \in \Act_\exists(\obs(l))$ of $\playere$, there exist an action $a \in \Act_\forall(\obs(l),a^\exists)$ and $l' \in B_i$ such that for every $u' \in l'$ it holds that $\pre_0(u',a) \cap l \not = \emptyset$,
 \item if $l \in \mc P(\mc P (S_1))$ then there exists $l' \in B_i$ such that for every $u' \in l'$ it holds that $\pre_1(u') \cap l \not = \emptyset$.
\end{itemize}

The ordering $\sqsubseteq$ on $\lsets(S)$ is defined such that for $l, l' \in \lsets(S)$, we have  $l \sqsubseteq l'$ iff for every $u \in l$ there exists a $u' \in l'$ such that $u \supseteq u'$. The ordering $\sqsubseteq$ is a BQO, since $\supseteq$ is a BQO on $\ucsets(S)$. Intuitively, if $l$ belongs to the set of elements of $\lsets(S)$ in which $\playere$ cannot win, so does every $l'$ with $l \sqsubseteq l'$.

We say that the sequence $B_0,B_1,B_2\ldots$ \emph{converges at $k$} if $\Min(B_{k+1}) \subseteq \Min(B_{k})$, where $\Min(B_{i})$ is the set of minimal  elements of $B_i$ w.r.t.\ $\sqsubseteq$. This condition can be effectively checked, since each $B_i$ is finite. We argue that there exists a $k \geq 0$ such that the sequence computed by the procedure described above converges at $k$ (and hence the procedure will terminate).

Let $F_0,F_1,F_2,\ldots$ be the sequence of upward-closed elements of $\mc P(\lsets(S))$ where $F_i = B_i\uparrow$ for each $i \geq 0$. As $F_0,F_1,F_2\ldots$ is a monotonically increasing sequence of upward-closed sets of elements of $\lsets(S)$, it must eventually stabilize, i.e., there is a $k \geq 0$ such that $F_{k+1} \subseteq F_k$. Thus, since $F_{i+1} \subseteq F_i$ if and only if $\Min(B_{i+1}) \subseteq \Min(B_{i})$, the sequence $B_0,B_1,B_2\ldots$ is guaranteed to converge at some $k \geq 0$.

\begin{proposition}
Let $B = B_{k}$, where the sequence $B_0,B_1,B_2\ldots$ converges at $k$. Then, $\playere$ has an $\widetilde\obs$-consistent winning strategy in $\safety(\games,\err)$ iff for every $l \in B$ there exists $u \in l$ with $u \cap I = \emptyset$. 

If $\playere$ has an $\widetilde\obs$-consistent winning strategy in $\safety(\games,\err)$, then $\playere$ has a finite-state $\obs^+$-consistent winning strategy in the original game $\safety(\mc G,\err)$.
\end{proposition}
\begin{proof}[Proof Idea]
A counterexample tree for $\safety(\games,\err)$ represents a witness for the fact that $\playere$ does not have an $\widetilde \obs$-consistent winning strategy in $\safety(\games,\err)$. It is a finite tree with nodes labeled with elements of $\dcsets(S)$. If there is a $l \in B$ such that $u \cap I \not= \emptyset$ for every $u \in l$, a counterexample tree can be constructed in a top-down manner. For the other direction we can show by induction on the depth of the existing counterexample trees that there exists a $l \in B$ such that $u \cap I \not= \emptyset$ for every $u \in l$.

For the case when $\playere$ wins the game $\safety(\games,\err)$ we can construct a finite-state $\obs^+$-consistent winning strategy for $\playere$ in the game $\safety(\mc G,\err)$ by using as states for the strategy automaton functions from observations to a finite set $\mc V \subseteq \pfin(\mc P (S))$ each of whose elements $V$ preserves the invariant that for every $l \in B$ there exists a $u \in l$ such that $u \cap \bigcup_{v \in V} v = \emptyset$. 
\end{proof}

%% file: reachability.tex
\myparagraph{LC-Games under incomplete information with reachability objectives} 
For reachability games we give a procedure based on forward exploration of the sets of states representing the knowledge of $\playere$ about the current state of the game. Since $\playere$ can only observe the heads the observable channels, his knowledge at each point of the play is a finite downward-closed set, element of $\dcsetsf(S)$. To update this knowledge we define functions $\post_0^\obs: \dcsetsf(S) \times \Sigma_0 \to \pfin(\dcsetsf(S))$ and $\post_1^\obs: \dcsetsf(S) \to \pfin(\dcsetsf(S))$ that map a set $d \in \dcsetsf(S)$ to a finite set of elements of $\dcsetsf(S)$, each of which is a set of states that $\playere$ knows, according to his current observation, the game may be in after (a transition from $\Sigma_0$ and) a sequence of transitions from $\Sigma_1$. For each $d \in \dcsetsf(S)$ we have $d' \in \post_0^\obs(d,a)$ (respectively $d' \in \post_1^\obs(d)$) iff there exists a sequence $d_0,d_1,\ldots,d_n \in\dcsetsf(S)$ such that 
$d_0 \in  \post_0(d,a)$ (respectively $d_0 = d$), for every $1 \leq i \leq n$ it holds that $d_{i-1} \subseteq S_1$ and $d_i \in \post_1(d_{i-1})$, and 
for every $0\leq i < j < n$ it holds that $d_i \not \subseteq d_j$
and one of the following conditions is satisfied: 
(1) $d' \subseteq \goal$, $d' = d_0$ and $n = 0$ (i.e., $d' \subseteq \goal \cap S_1$), or
(2) there exists a $1 \leq i < n$ such that $d_i \subseteq d_n$ and $d' = d_n$ (i.e., $d' \subseteq S_1$), or
(3) $d' = \{(0,q_0',q_1',w') \mid (1,q_0',q_1', w') \in \bigcup_{i = 0}^{n} d_i\}$ (i.e., $d' \subseteq S_0$).

We construct a finite set of trees rooted at the different possible knowledge sets for $\playere$ at location $q_o^0$.
The nodes of the trees are labeled with knowledge sets, i.e., with elements of $\dcsetsf(S)$. The edges are labeled wit pairs of transition labels, i.e., elements of $\actebot\times\Sigma_0$, where the first element of a pair is a possible choice of $\playere$ and the second element is a corresponding choice of $\playerf$. 

Formally, the forward exploration procedure constructs a forest $\mc T$ in which the roots are labeled with the sets $\{(0,q_0^0,q_1^0,\epsilon)\}$ and all the sets $d\in \post_1^\obs(\{(1,q_0^0,q_1^0,\epsilon)\}\setminus \goal)$. At each step of the construction an open leaf node $n$ with label $d$ is processed in the following way:
\begin{itemize} \setlength{\itemsep}{-.5mm}
 \item If $d \subseteq \goal$, we close the node and do not expand further from this node.
 \item If $d \not\subseteq \goal$ and either $d \subseteq S_0$ and there exists an ancestor of $n$ that is labeled with $d'$ and such that $d' \subseteq d$, or $d \subseteq S_1$, we close the node and do not expand further from this node.
\item Otherwise, we add the set of successors of $n$: for each $a^\exists \in \Act_\exists(\obs(d))$, each $a \in \Act_\forall(\obs(d),a^\exists)$ and each 
$d'\in\post_0^\obs(d,a)$ we add exactly one successor $n'$ labeled with $d'$ and label the edge $(n,n')$ with $(a^\exists,a)$. The set of successors for $(a^\exists,a)$ is denoted with $\Children(n,a^\exists,a)$. 
\end{itemize}

The finite branching of the transition relation of $\mc G$ and the fact that $\subseteq$ is a BQO on $\dcsetsf(S)$ imply that each of the sets $\post_0^\obs(d,a)$ and $\post_1^\obs(d)$ can be effectively computed, the set of roots and the out-degree of each node are finite, and the above procedure terminates constructing a finite forest $\mc T$. 

We label each node $n$ in $\mc T$ with a boolean value $\win(n)$. For a leaf node $n$ with $d(n) \subseteq \goal$, we define $\win(n) = \true$ and for any other other leaf node $n$ we define $\win(n) = \false$. The value of a non-leaf node is computed based on those of its children by interpreting the choices of $\playere$ disjunctively and the choices of $\playerf$ conjunctively. Formally, for every non-leaf node $n$ we define $\win(n) = \bigvee_{a^\exists \in \Act_\exists(\obs(d(n)))}\bigwedge_{a \in \Act_\forall(\obs(d(n)),a^\exists)}\bigwedge_{n' \in \Children(n,a^\exists,a)} \win(n')$, where $d(n)$ is the set of states labeling $n$.

\smallskip

\begin{proposition}
$\playere$ has an $\obs^+$-consistent winning strategy in $\reach(\mc G,\goal)$ iff for every root $n$ in $\mc T$ it holds that $\win(n) = \true$.
If $\playere$ has an $\obs^+$-consistent winning strategy in $\reach(\mc G,\goal)$, then he also has a finite state $\obs^+$-consistent winning strategy in $\reach(\mc G,\goal)$. 
\end{proposition}
\begin{proof}[Proof Idea]
If all the roots are labeled with $\true$ we can construct a finite-state $\obs^+$-consistent strategy winning for $\playere$ in $\reach(\mc G,\goal)$, by mapping each prefix in $\prefse(\mc G)$ to a label in $\actebot$, determined by a corresponding path in $\mc T$ and a fixed successful choice at its last node, if such path and choice exist, or given an appropriate default value otherwise.
For the other direction we suppose that some root is labeled with $\false$ and show that for any $\obs^+$-consistent strategy $\stre$ for $\playere$, we can use the tree to construct a play $\pi \in \outcome(\stre)$ that never visits a state in $\goal$.
\end{proof}

%% file: parity.tex
\myparagraph{LC-games under incomplete information with parity objectives} We now turn to more general $\omega$-regular visible objectives for $\playere$ where the undecidability results established in~\cite{AbdullaBd/08/MonotonicAndDownwardClosedGames} for perfect information lossy channel games in which only one player can lose messages, carry on to our setting.

A \emph{visible priority function} $\pr : \Obs \to \{0,1,\ldots,n\}$ for natural number $n \in \mathbb N$ maps each observation to a non-negative integer priority. For an infinite play $\pi = s_0 a_0^\exists a_0 s_1 a_1^\exists a_1\ldots$ we define $\pr(\pi) = \min \{\pr(o)\mid o \in \InfObs(\pi)\}$, where $\InfObs(\pi)$ is the set of observations that occur infinitely often in $\pi$, and define $\wpr(\pi) = \min \{\pr(\obs(s_0)),\pr(\obs(s_1)),\ldots\}$.
A \emph{parity} (respectively \emph{weak parity}) LC-game under incomplete information $\parity(\mc G,\pr)$ (respectively $\wparity(\mc G,\pr)$) is defined by a LC-game structure with incomplete information $\mc G$ and a visible priority function $\pr$. 
A strategy $\stre$ for $\playere$ is \emph{winning} in the parity game $\parity(\mc G,\pr)$ (weak parity game $\wparity(\mc G,\pr)$) iff for every infinite play $\pi \in \outcome(\stre)$ it holds that $\pr(\pi)$ is even (respectively $\wpr(\pi)$ is even).

\smallskip

\begin{proposition}
The weak parity game solving problem for LC-games under incomplete information, that is, given a weak parity LC-game under incomplete information $\wparity(\mc G,\pr)$ to determine whether there exists an $\obs^+$-consistent winning strategy for $\playere$ in $\wparity(\mc G,\pr)$, is undecidable.
\end{proposition}
\begin{proof}[Proof Idea]
In~\cite{AbdullaBd/08/MonotonicAndDownwardClosedGames} it was shown that in the perfect information setting the weak parity problem for B-LCS games, which are games played on a finite set of channels in which player A has a weak parity objective and only player B is allowed to lose messages, is undecidable. Their proof (given for A-LCS games but easily transferable into a proof for B-LCS games) is based on a reduction from the infinite computation problem for transition systems based on lossy channel systems, which is undecidable~\cite{AbdullaJ/96/UndecidableVerificationProblemsForProgramsWithUnreliableChannels}.

We argue that this reduction can be adapted for our framework, with $\playere$ in the role of player A and $\playerf$ in the role of player B. The fact that here $\playere$ choses only transition labels and plays under incomplete information does not affect the proof for B-LCS games, since there player A just follows passively, while player B simulates the original system. The values of the priority function used in~\cite{AbdullaBd/08/MonotonicAndDownwardClosedGames} do not depend on the contents of the channels. Thus, we can define a visible priority function.
\end{proof}

As a consequence, the parity game solving problem for LC-games under incomplete information is undecidable as well. 
As noted in~\cite{AbdullaBd/08/MonotonicAndDownwardClosedGames}, the construction from the proposition above can be used to show undecidability of A-LCS and B-LCS games with B\"uchi and co-B\"uchi objectives. 

\bigskip

\myparagraph{Summary of the results} The results of the paper are summarized in the following theorem.
\begin{theorem}
For lossy channel game structures with incomplete information
\begin{itemize}\setlength{\itemsep}{-.5mm} 
 \item games with visible safety or reachability objectives for $\playere$ are decidable, and when $\playere$ has an observation-based winning strategy, a finite-state such strategy can be effectively computed, 
 \item games with visible weak parity objectives for $\playere$ are undecidable.
\end{itemize}
\end{theorem}

%% file: conclusion.tex
We showed that the game solving problem for LC-games under incomplete information with safety or reachability objective for $\playere$ is decidable. LC-games under incomplete information with more general winning conditions, such as weak parity (as well as B\"uchi and co-B\"uchi) condition can easily be shown to be undecidable, using a reduction similar to the one described in~\cite{AbdullaBd/08/MonotonicAndDownwardClosedGames} for A-LCS games (which are perfect information games defined on LCSs in which only one player can lose channel messages).
An orthogonal extension that is also clearly undecidable is decentralized control. This implies that suitable abstraction techniques are needed to address the synthesis problem within these undecidable settings.   